\documentclass[12pt]{article}
 
\usepackage{microtype}

\usepackage{amsmath,amssymb,amsthm,latexsym}
\usepackage{fullpage}
\usepackage{color}
\usepackage[linesnumbered,ruled,vlined]{algorithm2e}

\def\rk{{\rm rank}}
\def\cork{\mbox{\rm corank}}
\def\mincork{\cork}
\newcommand{\disc}{\mathrm{disc}}

\def\Q{{\mathbb Q}}
\def\R{{\mathbb R}}
\def\C{{\mathbb C}}
\def\F{{\mathbb F}}
\def\N{{\mathbb N}}

\def\Hom{{\rm Hom}}

\def\Elin{{\rm Lin}}

\def\rem#1{}

\def\cH{{\cal H}}

\def\cA{{\cal A}}
\def\cB{{\cal B}}
\def\cD{{\cal D}}
\def\cT{{\cal T}}
\def\matsp{\Elin(V, V')}
\def\matspc{M(n, \F)}
\def\matspe{M(n, \F')}

\def\vecspc{\F^n}

\newcommand{\fail}{\mathrm{{\texttt{Fail}}}}

\newcommand{\trialgo}{\textsc{TriAlgo}}

\newcommand{\maxrk}{\mathrm{maxrk}}

\newcommand{\sk}{\mathrm{sk}}

\def\rmitem[#1]#2{\item[{\rm #1}]#2}
\newcommand{\comment}[1]{}
\newcommand{\im}{\mathrm{im}}
\newcommand{\pep}{power overflow problem}
\newcommand{\tp}[1]{{#1}^{\mathrm{T}}}
\newcommand{\class}[1]{\mathbf{#1}}
\newcommand{\rkone}{\class{R_1}}
\newcommand{\ut}{\class{UT}}

\newtheorem{theorem}{Theorem}
\newtheorem{proposition}[theorem]{Proposition}
\newtheorem{lemma}[theorem]{Lemma}
\newtheorem{corollary}[theorem]{Corollary}
\newtheorem{fact}[theorem]{Fact}
\newtheorem{definition}[theorem]{Definition}
\newtheorem{problem}[theorem]{Problem}

\newtheorem{remark}[theorem]{Remark}

\title{
Generalized Wong sequences \\
and their applications to Edmonds' problems
}

\author{
G\'abor Ivanyos\thanks{Institute for Computer Science and Control, Hungarian 
Academy of Sciences, 
Budapest, Hungary 
({\tt Gabor.Ivanyos@sztaki.mta.hu}).} 
\and 
Marek Karpinski\thanks{Department of Computer Science, 
University of Bonn, Bonn, Germany
({\tt marek@cs.uni-bonn.de}).}
\and
 Youming Qiao\thanks{Centre for Quantum Computation and Intelligent Systems
 University of Technology, Sydney; and Centre for Quantum Technologies,
 National University of Singapore, Singapore 117543. 
 ({\tt jimmyqiao86@gmail.com})}
\and
 Miklos Santha\thanks{LIAFA, Univ. Paris 7, CNRS, 75205 Paris, France;  and 
Centre for Quantum Technologies, National University of Singapore, 
Singapore 117543 ({\tt miklos.santha@liafa.jussieu.fr}).}
 }
 
\date{}

\begin{document}

\maketitle

\begin{abstract}

We design two deterministic polynomial-time algorithms for variants of a 
problem introduced by Edmonds in 1967: 
determine the rank of a matrix $M$ whose entries are homogeneous linear polynomials
over the integers.
Given a linear subspace $\cB$ of the $n \times n$ matrices over some field $\F$, we consider the following problems:
\emph{symbolic matrix rank} (SMR) is the problem to 
determine the maximum rank among
matrices in $\cB$, while 
{\em symbolic determinant 
identity testing} (SDIT) is 
the question to decide whether there exists a nonsingular matrix 
in $\cB$. The constructive versions of these problems are asking to find a matrix of maximum rank, respectively a 
nonsingular matrix, if there exists one.

Our first algorithm solves the {\em constructive} SMR when $\cB$ is spanned by unknown rank one matrices, answering
an open question of Gurvits. Our second algorithm solves the constructive SDIT when $\cB$ is spanned by 
triangularizable matrices, but the triangularization is not given explicitly. 
Both algorithms work 
over 
fields of size at least $n+1$, 
and the first algorithm 
actually solves (the non-constructive)
SMR independent of the field size.
Our framework is based on a generalization of Wong sequences,
a classical 
method to deal with pairs of matrices, 
to the case of pairs of matrix spaces. 
\footnote{
A preliminary report on this work appeared in~\cite{conf_version}.
}
 \end{abstract}
 
\section{Introduction}\label{sec:intro}

In 1967, Edmonds introduced the following problem \cite{Edmonds}:
Given a matrix $M$ whose entries are homogeneous linear polynomials
over the integers, determine the rank of $M$. The problem is
the same as determining the maximum rank of a matrix in 
a linear space of matrices over the rationals. In this paper
we consider this question and its certain variants over
more general fields. 

Let us denote by $\matspc$ the linear space of $n\times n$ 
matrices over a field $\F$. 
We call a linear subspace $\cB \leq \matspc$ a \emph{matrix space}. We define the
\emph{symbolic matrix rank} problem (SMR)
over $\F$ as follows: given 
$\{B_1, \dots, B_m\}\subseteq \matspc$, determine the maximum rank among
matrices in $\cB = \langle B_1, \dots, B_m\rangle$, 
the matrix space spanned by $B_i$'s.
The \emph{constructive} version of SMR is to find a matrix of maximum rank in $\cB$
(this is called the maximum rank matrix completion problem in \cite{Geelen} 
and in \cite{IKS}).
We refer to the weakening of SMR,
when the question is to decide whether there exists a nonsingular matrix 
in $\cB$, as the  {\em symbolic determinant 
identity testing} problem (SDIT), the name used 
by~\cite{KI04} (in~\cite{Gurvits} this variant is called Edmonds' problem). 
The \emph{constructive} version in that case is to find a nonsingular matrix, if there is one in $\cB$.
We will occasionally refer to any of the above problems as \emph{Edmonds' problem}.



The complexity of the SDIT
depends crucially on the size of the underlying field $\F$. When $|\F|$ 
is a constant then it is NP-hard \cite{BFS}. On the other hand if 
the field size is large enough (say $\geq 2n$) then by the
Schwartz-Zippel lemma \cite{Schwartz,Zippel} 
it admits an efficient randomized algorithm \cite{Lovasz79}. 
Obtaining a deterministic polynomial-time algorithm for the SDIT
would be of fundamental importance, since
Kabanets and Impagliazzo~\cite{KI04}  showed 
that such an algorithm would imply strong circuit lower 
bounds which seem beyond current techniques. 

Previous works on Edmonds' problems mostly dealt with the case when the 
given \emph{matrices} $B_1, \dots, B_m$ satisfy certain property. 
For example, Lov\'asz~\cite{Lovasz}  considered several cases of SMR, including 
when the $B_i$'s are of rank $1$, and when they are skew symmetric 
matrices of rank $2$. These classes were then shown to have 
deterministic polynomial-time algorithms 
\cite{Geelen,Murota,HKM05,GIM03,GI05,IKS},
see Section~\ref{subsec:previous} for more details.

Another direction also studied is when instead of the given matrices, the 
spanned \emph{matrix space} $\cB = \langle B_1, \dots, B_m\rangle$
satisfies certain property. Since such a property is just a subset of all 
matrix spaces, we also call it a class of matrix spaces. 
Gurvits~\cite{Gurvits} presented an efficient 
deterministic algorithm for the SDIT
over 
subfields of $\C$,
when the matrix space falls in a special class,
what we call the \emph{Edmonds-Rado class}. 
We shall review the definition of this class 
in Section~\ref{subsec:previous}. 
In this paper, our main
goal is to consider Edmonds' problems for the following two 
classes.
\begin{itemize}
\item The class of rank-1 spanned matrix spaces, $\rkone$: a matrix space 
$\cB\leq \matspc$ is in $\rkone$, if $\cB$ has a basis consisting of rank-$1$ 
matrices over $\F'$, where $\F'$ is some extension field of $\F$.\footnote{Note 
that it is possible for $\cB$ to have a rank-$1$ basis over $\F'$ but no such 
over $\F$. 
See \cite{Gurvits_matching} for an example. }
\item The class of (upper) triangularizable matrix spaces, $\ut$: a matrix 
$\cB\leq \matspc$ is in $\ut$, if there exist nonsingular $C, D\in M(n, \F')$, 
where $\F'$ is some extension field of $\F$, such that
for all $B\in\cB$, the matrix
$DBC^{-1}$ is upper-triangular.
\end{itemize}
It is known that the
Edmonds-Rado class 
includes $\rkone$ and $\ut$. See 
Section~\ref{subsec:previous} for more details. While Gurvits presented an 
efficient deterministic SDIT algorithm for the Edmonds-Rado class 
over 
subfields of $\C$,
the same problem over (large enough) finite fields is still open, even for 
special classes like $\rkone$ and $\ut$. In fact, 
Gurvits stated as an open question the complexity of
the SMR for $\mathbf{R_1}$ over 
finite fields \cite[page 456]{Gurvits}. 

The difference between 
properties of matrices and properties of matrix spaces is 
critical for Edmonds' problems. 
In particular, whether a matrix space satisfies a certain property or not, 
should not depend on choices of basis. 
We are not aware of any result on the complexity 
of finding rank one generators for a subspace
$\cB$ in $\rkone$ if it is given by a basis
consisting of not necessarily rank one matrices.
We believe that the problem is hard.
Thus the existence of algorithms for SMR when the $B_i$'s are rank-$1$ 
does not immediately imply algorithms for matrix spaces in $\mathbf{R_1}$. 

Furthermore, most properties we encounter in practice respect the following 
equivalence relation of matrix spaces. Two matrix spaces $\cA$ and $\cB$ in 
$\matspc$ are 
equivalent, if there exist nonsingular $C, D\in M(n, \F)$, s.t. $\cA=C\cB 
D:=\{CBD\mid B\in \cB\}$. 
Edmonds-Rado 
class, 
$\rkone$ and $\ut$ all respect 
this equivalence relation. Again, given matrices $B_1, \dots, B_m$, and suppose 
$\cB=\langle B_1, \dots, B_m\rangle$ is in $\ut$, 
it is not clear how difficult is computing
matrices $C, D$ that triangularize 
$\cB$. 
The problem does not look as hard as 
finding rank one generators, see Section~\ref{sec:concl} for
some details.
Thus while SDIT for upper-triangular $B_i$'s is easy, it does not immediately 
suggest an algorithm for matrix spaces in $\ut$. 

To ease the description of our results, we make a few definitions and notations. 
We denote by $\rk(B)$ the rank of a matrix $B$, and we set $\cork(B) = n - \rk(B).$ 
For a matrix space $\cB$ we set $\rk(\cB)=\max\{\rk(B)\mid B\in\cB\}$ and $\cork(\cB)=n-\rk(\cB)$.
We say that $\cB$ is \emph{singular} if $\rk(\cB) < n$, that is 
if $\cB$ does not contain a nonsingular element, and \emph{nonsingular} otherwise. 

For a subspace $U\leq \vecspc$, we set $\cB(U)=\langle B(u)\mid B\in \cB, u\in U\rangle$.
Let $c$ be a nonnegative integer. We say that $U$ is a 
\emph{$c$-singularity witness of $\cB$}, if $\dim(U)-\dim(\cB(U))\geq c$, and 
$U$ is a \emph{singularity witness} of $\cB$ if for some $c>0$, it is a 
$c$-singularity witness.
Note that if there exists a singularity witness of $\cB$ then $\cB$ can only be singular. Let us define the \emph{discrepancy} of $\cB$ 
as $\disc(\cB)=\max\{c\in \N\mid \exists ~c$-singularity witness of $\cB\}$. Then it is also clear that 
$
\cork(\cB)\geq \disc(\cB).
$

Our main results 
are algorithms that run in polynomial time on
an {\em algebraic RAM}~\cite{kalt88}, a random access machine in
which the field operations as well as testing equality of field
elements are performed at unit cost. 
Over finite fields, the straightforward implementations
of these algorithms automatically have polynomial (in $\log \F$ and $n$)
Boolean (``bit'') complexity. With some
effort, we are also able to present deterministic algorithms over the
rationals which have Boolean complexity polynomial 
in the number of bits representing the input data.
We now state our main theorems.

\begin{theorem}\label{thm:main1}
There are deterministic algorithms which solve
the SMR 
on an algebraic RAM for $\F$ or over $\Q$, respectively,
in polynomial time if $\cB$ 
is spanned by rank-1 matrices. 
If the size of the base field is at least $n+1$,
the algorithm solves the constructive SMR, and it also outputs
a $\cork(\cB)$-singularity witness.
\end{theorem} 

\begin{theorem}\label{thm:main2}
Assume that the size of the base field $\F$ is at least $n+1$.
Then there are deterministic polynomial-time
algorithms which solve the constructive SDIT 
on an algebraic RAM for $\F$ or over $\Q$, respectively,
if $\cB$ is triangularizable.
Furthermore,  
when $\cB$ is singular,
the algebraic RAM algorithm 
also outputs a singularity witness.
\end{theorem}



Theorem~\ref{thm:main1} can be slightly strengthened as follows: instead of 
assuming
that the whole space $\cB$ is rank-$1$ spanned, it is 
sufficient to suppose that a subspace of $\cB$ of co-dimension one is 
spanned by rank-$1$ matrices. See Remark \ref{remark:thm1} (2) for the work 
needed to achieve this. 

Let us comment briefly on the
framework for our algorithms. 
We generalize the first and second Wong sequences
for matrix pencils (essentially two-dimensional matrix spaces)
which have turned out to be useful among others 
in the area of linear 
differential-algebraic equations (see the recent survey \cite{Tre13}).
These were originally defined in \cite{Wong} for a pair of 
matrices $(A, B)$, and were recently used to compute the 
Kronecker normal form in a numerical stable way \cite{BT12,BT12addition}. 
We generalize Wong sequences to the 
case $(\cA, \cB)$ where $\cA$ and $\cB$ are matrix spaces, 
and show that they have analogous basic properties to the original ones. We 
relate the generalized Wong sequences to
Edmonds' problems via singularity witnesses. Essentially this connection allows 
us to design the algorithm 
for  $\mathbf{R_1}$ using the second Wong sequence, and the algorithm for 
$\ut$ using the first Wong sequence.
We remark that the application of the second Wong sequence is not
new. Similar techniques were used in \cite{IKS} to find
maximum rank matrices in the case where rank one generators
for $\cB$ were given.
Furthermore, while preparing the present version, 
we became aware of the paper \cite{fr04} by Fortin and Reutenauer
in which essentially the same method is used for testing existence of 
$\cork(\cB)$-singularity witnesses (on a randomized algebraic RAM).



\subsection{Comparison with previous works}\label{subsec:previous}


The idea of singularity witnesses was already present in Lov\'asz's work \cite{Lovasz}. 
Among other things, Lov\'asz showed that for the rank-$1$ spanned case, the 
equality 
$
\cork(\cB) = \disc(\cB)
$
holds, by reducing it to Edmonds' Matroid Intersection theorem \cite{Edmonds70}, which in turn can be deduced from Rado's matroidal generalization of Hall's theorem \cite{Rado} (see also \cite{Welsh}).
Inspired by this fact, Gurvits introduced the term
\emph{Edmonds-Rado property} for 
membership in
the class of matrix  spaces which are
either nonsingular, or have a singularity witness.
Throughout this paper we refer to this class as
the \emph{Edmonds-Rado class}.
Gurvits listed several subclasses of the Edmonds-Rado class, 
including $\mathbf{R_1}$ (by the aforementioned result of Lov\'asz) and $\ut$.
A well-known example of a matrix space 
outside the Edmonds-Rado class 
is the linear space of skew symmetric matrices of size $3$ 
\cite{Lovasz}. 

As we stated already, Gurvits has  presented a polynomial-time deterministic 
algorithm for the SDIT over
subfields of $\C$ 
for matrix spaces in the Edmonds-Rado class. Therefore 
over 
these fields,
his algorithm covers the SDIT for 
$\mathbf{R_1}$ and for $\ut$. 
Our algorithms 
(in the algebraic RAM model)
are valid 
over arbitrary sufficiently large fields.
In the triangularizable case we also deal with
the SDIT, but for $\mathbf{R_1}$ we solve the more general SMR. 
In fact, it is not hard to reduce SMR for the general to SMR for the 
triangularizable case (see Lemma~\ref{prop:reduction_to_triangular}), so 
solving SMR for $\ut$ is as hard as the general case.
In both cases the algorithms solve the constructive version of the problems, 
and they also construct singularity witnesses.
Finally, they work
in polynomial time when the field size is at least $n+1$. Moreover, for 
$\mathbf{R_1}$ the algorithm solves the non-constructive SMR in polynomial time 
even over arbitrarily small finite fields,
settling 
an
open problem of Gurvits.


Over fields of constant size, the SMR has certain practical implications \cite{HKM05,HKY06},
but is shown to be NP-hard \cite{BFS} in general.
Some special cases have been studied, mostly in the form of the \emph{mixed matrices}, that is linear matrices where each entry is either a variable or a field element. Then by restricting the way variables appear in the matrices some cases turn out to have efficient deterministic algorithms, including when every variable appears at most once (\cite{HKM05}, building on \cite{Geelen,Murota}), 
and when the mixed matrix is skew-symmetric and every variable appears at most twice (\cite{GIM03,GI05}). Finally in \cite{IKS}, Ivanyos, Karpinski and Saxena present a deterministic polynomial-time algorithm for the case when 
among the input matrices $B_1, \dots, B_m$ all but $B_1$ are of rank $1$. 


As a computational model of polynomials, determinants with affine polynomial 
entries turn out to be equivalent to algebraic branching programs 
(ABPs) \cite{Valiant,Berkowitz} up to a polynomial overhead. 
Thus the identity test for ABPs is the same as SDIT. 
For restricted classes of ABPs, (quasi)polynomial-time deterministic 
identity test algorithms have been devised 
(cf. \cite{FS12} and the references therein). 
Note that identity test results for SDIT and ABPs are 
in general incomparable. For an application
of SDIT to quantum information processing see \cite{CDS10}.


\paragraph{Organization.} In Section~\ref{sec:wong} we define Wong sequences of 
a pair of matrix spaces, and present their basic properties. In 
Section~\ref{sec:second} the connection between the second Wong sequence and 
singularity witnesses is shown. Based on this connection we introduce the 
\pep{}, and reduce the SMR  to it. We also prove here Theorem~\ref{thm:main1}
under the hypothesis that there is a polynomial time algorithm for the \pep. In 
Section~\ref{sec:rank1} we show an algorithm for the \pep{} that works in 
polynomial time for rank-$1$ spanned matrix spaces. Section~\ref{sec:first} is 
devoted to the algorithm for triangularizable matrix spaces, proving 
Theorem~\ref{thm:main2}. 
Finally, in Section~\ref{sec:more_ER} we propose and investigate some natural 
subclasses of the Edmonds-Rado class. 






\section{Wong sequences for pairs of matrix spaces}\label{sec:wong}

For $n\in\N$, we set $[n]=\{1, \dots, n\}$. We use $0$ to denote the zero vector space. 
In this section we generalize the classical Wong sequences of matrix pencils to 
the situation of pairs of matrix subspaces. 
This is 
the framework for the algorithms in 
this work. Let $V$ and $V'$ be finite dimensional vector spaces over a field 
$\F$, and let $\matsp$ be the vector space of linear maps from $V$ to $V'$. 
Suppose $n=\dim(V)$ and $n'=\dim(V')$. 

Let $U\leq V$ and $W\leq V'$ be subspaces of $V$ and $V'$, respectively. For 
$A\in\matsp$, the image of $U$ under $A$ is $A(U)=\{A(u)\mid u\in U\}$, and the 
preimage of $W$ under $A$ is $A^{-1}(W)= \{v\in V\mid A(v)\in W\}$. To define 
generalized 
Wong sequences, the first step is to generalize the definitions of image and 
preimage under a single linear map $A$, to those under a matrix space 
$\cA\leq \matsp$. 

Naturally, the image of $U$ under $\cA$ is the span of the 
images of $U$ under every $A\in \cA$, that is $\cA(U)=\langle 
\cup_{A\in\cA}A(U)\rangle=\langle \{A(u)\mid A\in 
\cA, u\in U\}\rangle.$
On the other hand, the preimage of $W$ under $\cA$ may be somewhat unexpected. 
It turns out that we 
need to take the intersection of the preimages of $W$ under every $A\in \cA$, 
that is $\cA^{-1}(W)=\cap_{A\in\cA}A^{-1}(W)=\{v\in V\mid \forall A\in \cA, 
A(v)\in W\}$.
Note that $\cA(U)$ (resp. $\cA^{-1}(W)$) is a subspace of $V'$ (resp. $V$). 
Moreover, if
$\cA$ is spanned by $\{ A_1, \dots, A_m\}$, then $\cA(U)=\langle \cup_{i\in[m]} A_i(U)\rangle$, and $\cA^{-1}(W)=\cap_{i\in[m]} A_i^{-1}(W)$.
Some easy and useful facts are the following. 
\begin{lemma}\label{lemma:contain}
For $\cA, \cB\leq \matsp$, and $U, S\leq V$, $W, T\leq V'$, we have:
\begin{enumerate}
\item If $U\subseteq S$ and $W\subseteq T$, then $\cA(U)\subseteq \cA(S)$ and $\cA^{-1}(W)\subseteq \cA^{-1}(T)$; 
\item If $\cB(U)\subseteq \cA(U)$ and $\cB(S)\subseteq \cA(S)$, then $\cB(\langle U\cup S\rangle)\subseteq \cA(\langle U\cup S\rangle)$; 
\item If $\cB^{-1}(W)\supseteq \cA^{-1}(W)$ and $\cB^{-1}(T)\supseteq \cA^{-1}(T)$, then $\cB^{-1}(W\cap T)\supseteq \cA^{-1}(W\cap T)$;
\item $\cA^{-1}(\cA(U))\supseteq U$, and $\cA(\cA^{-1}(W))\subseteq W$. 
\end{enumerate}
\end{lemma}

We now define two Wong sequences for a pair of matrix subspaces. 
\begin{definition}\label{def:wong_seqs}
Let $\cA, \cB\leq \matsp$. The sequence of subspaces $(U_i)_{i\in\N}$ of $V$  is called the \emph{first Wong sequence of $(\cA, \cB)$},
where $U_0=V$, and $U_{i+1}=\cB^{-1}(\cA(U_i))$. 
The sequence of subspaces $(W_i)_{i\in\N}$ of $V'$ is called the \emph{second Wong sequences of $(\cA, \cB)$}, where
$W_0=0$, and $W_{i+1}=\cB(\cA^{-1}(W_i))$. 

\end{definition}

When $\cA=\langle A\rangle$ and $\cB=\langle B\rangle$ are one 
dimensional matrix spaces, 
the Wong sequences for $(\cA,\cB)$ coincide with
the classical Wong sequences for the matrix pencil $Ax-B$
\cite{Wong,BT12}. 
%
The following properties are straightforward generalizations of those for classical Wong sequences. We 
start by considering the first Wong sequence. 

\begin{proposition}\label{prop:first}
Let $(U_i)_{i\in\N}$ be the
first Wong sequence of $(\cA, \cB)$. Then for all $i\in\N$, we have $U_{i+1}\subseteq U_i$. Furthermore, $U_{i+1}=U_i$ if and only if $\cB(U_i)\subseteq \cA(U_i)$. 
\end{proposition}
\begin{proof}
Firstly we show that $U_{i+1}\subseteq U_i$, for every $i\in\N$. For $i=0$, this holds trivially. For $i>0$, by 
Lemma \ref{lemma:contain} (1) we get $U_{i+1}=\cB^{-1}(\cA(U_i))\subseteq \cB^{-1}(\cA(U_{i-1}))=U_{i}$, since $U_{i}\subseteq U_{i-1}$.


Suppose now that $\cB(U_i)\subseteq \cA(U_i)$, for some $i$. 
Then $U_i \subseteq \cB^{-1}(\cB(U_i)) \subseteq \cB^{-1}(\cA(U_i))$
respectively by Lemma~\ref{lemma:contain} (4) and (1),  which gives $U_{i+1}=U_i$. If $\cB(U_i)\not\subseteq \cA(U_i)$ then there exist $B\in\cB$ and $v\in U_i$ such that $B(v)\not\in \cA(U_i)$. Thus $v\not\in \cB^{-1}(\cA(U_i))=U_{i+1}$, which gives $U_{i+1}\subset U_i$. 
\end{proof}

Given Proposition~\ref{prop:first}, we see that the first Wong sequence stabilizes after at most $n$ steps at some subspace. That is, for any $(\cA, \cB)$, there exists $\ell\in\{0, \dots, n\}$, such that 
$U_0\supset U_1\supset \dots \supset U_\ell= U_{\ell+1}=\cdots$. 
In this case we call the subspace $U_\ell$ the \emph{limit} of $(U_i)_{i\in\N}$, and we denote it by $U^*$. 

\begin{proposition}\label{prop:largest}
$U^*$ is the largest subspace $T\leq V$ such that $\cB(T)\subseteq \cA(T)$. 
\end{proposition}
\begin{proof}
By Proposition~\ref{prop:first} we know that $U^*$ satisfies 
$\cB(U^*)\subseteq\cA(U^*)$.
Consider an arbitrary $T\leq V$ such that $\cB(T)\subseteq \cA(T)$, we show 
by induction that $T\subseteq U_i$, for all $i$. When $i=0$ this trivially holds. Suppose that 
$T\subseteq U_i$, for some $i$. Then by repeated applications of 
Lemma~\ref{lemma:contain} we have
$T\subseteq \cB^{-1}(\cB(T))\subseteq \cB^{-1}(\cA(T))\subseteq \cB^{-1}(\cA(U_i)) = U_{i+1}$. 
\end{proof}

Analogous properties hold for the second Wong sequence $(W_i)_{i\in\N}$. In particular the sequence stabilizes after at most $n'$ steps, and there exists a \emph{limit} subspace $W^*$ of $(W_i)_{i\in\N}$. We summarize them in the following proposition.

\begin{proposition}\label{prop:second}
Let $(W_i)_{i\in\N}$ be the
second Wong sequence of $(\cA, \cB)$.
Then
\begin{enumerate}
\item$W_{i+1}\supseteq W_i$,  for all $i\in\N$. Furthermore, $W_{i+1}=W_i$ if and only if $\cB^{-1}(W_i)\supseteq \cA^{-1}(W_i)$.
\item The limit subspace $W^*$ is the smallest subspace $T\leq V'$ s.t. $\cB^{-1}(T)\supseteq \cA^{-1}(T)$. 
\end{enumerate}
\end{proposition}

It is worth noting that the second Wong sequence can be viewed as the 
dual of the first one in the following sense. 
Assume that $V$ and $V'$ are equipped with nonsingular symmetric
bilinear forms, both denoted by $\langle , \rangle$. For a 
linear map  $A:V\to V'$ let $\tp{A}:V'\to V$ stand for the transpose of $A$ with respect
to $\langle , \rangle$. This is the unique map with the property
$\langle\tp{A}(u), v\rangle=\langle u, A(v)\rangle $, for all $u\in V'$ and $v\in V$.
For a matrix space $\cA$, let $\tp{\cA}$ be the space
$\{\tp{A}|A\in\cA\}$. 
For $U\leq V$, the orthogonal subspace of $U$ is defined as 
$U^\perp=\{v\in V\mid \langle v, u \rangle=0$ for all $u \in U\}$. 
Similarly we define $W^\perp$ for $W\leq V'$. 
Then we have $((\tp{\cA})^{-1}(U))^\perp=\cA(U^\perp)$, 
and $(\tp{\cA}(V))^\perp=\cA^{-1}(V^\perp)$. 
It can be verified that if $(W_i)_{i\in\N}$ is the second Wong sequence 
of $(\cA, \cB)$ and $(U_i)_{i\in\N}$ the first Wong sequence of 
$(\tp{\cA}, \tp{\cB})$, then $W_i=U_i^\perp$. We note  that the duality 
of Wong sequences was 
already derived
in~\cite{BT12} for pairs of matrices.




For a matrix space $\cA$ and a subspace $U\leq V$ given in terms of a basis we can compute
$\cA(U)$ by applying the basis elements for $\cA$ to those
of $U$ and then selecting a maximal set of linearly
independent vectors. A possible way of computing
$\cA^{-1}(U)$ for $U\leq V'$ is to compute first $U^\perp$,
then $\tp{\cA}(U^\perp)$ and finally
$\cA^{-1}(U)=(\tp{\cA}(U^\perp))^\perp$.
Therefore we have
\begin{proposition}\label{prop:wongcomplexity}
Wong sequences can be computed in time
using $(n+n')^{O(1)}$ 
on an algebraic RAM.
\end{proposition}


Unfortunately, we are unable to prove that over the rationals
the bit length of the entries of the bases describing 
the Wong sequences remain polynomially bounded in the length of 
the data for $\cA$ and $\cB$. 
However, in Section~\ref{subsec:second_connection} we show that if $\cA=\langle A\rangle$,
then the first few members of the second Wong sequence 
which happen to be
contained in $\im(A)$ can be computed in polynomial time using an iteration
of multiplying vectors by matrices from a basis for 
$\cB$ and by a 
pseudo-inverse
of $A$. 

We also observe that if we consider the bases
for $\cA$ and $\cB$ as matrices over an extension
field $\F'$ of $\F$ then the members of the Wong sequences
over $\F'$ are just the $\F'$-linear spaces spanned by
the corresponding members of the Wong sequences over $\F$.
In particular, the limit of the first Wong sequence
over $\F$ is nontrivial if and only if the limit of the
first Wong sequence over $\F'$ is nontrivial.

\section{The second Wong sequence and rank-$1$ spanned matrix spaces}
\label{sec:second}

\subsection{Second Wong sequences and singularity 
witnesses}\label{subsec:second_connection}

As in Section~\ref{sec:wong}, let $V$ and $V'$ be finite dimensional vector spaces over a field $\F$,
of respective dimensions $n$ and $n'$. For $A\in\matsp$ we set 
$\cork(A)=\dim(\ker(A))$. For $\cB\leq\matsp$, the concepts of $c$-singularity 
witnesses, $\disc(\cB)$ and $\cork(\cB)$, defined for the case when $n=n'$, can 
be generalized naturally to $\cB$. We also have that $\cork(\cB)\geq 
\disc(\cB)$, and that a $\cork(\cB)$-singularity witness of $\cB$ does not 
exist necessarily.
Let $A\in\cB$, and consider $(W_i)_{i\in\N}$, the second Wong sequence of $(A, \cB)$.
The next lemma states that
the limit $W^*$ 
is basically such a witness under the
condition that it is contained in the image of $A$. Moreover, in this specific case the limit can be computed efficiently.


\begin{lemma}\label{lem:second_utility}
Let $A\in\cB\leq\matsp$, and let $W^*$ be the limit of the second 
Wong sequence of $(A, \cB)$. There exists a  $\cork(A)$-singularity
witness of $\cB$ if and only if $W^*\subseteq \im(A)$.
If this is the case, then $A$ is of maximum rank and $A^{-1}(W^*)$ is a  $\cork(\cB)$-singularity
witness.
\end{lemma}
\begin{proof}
We  prove the equivalence. Firstly suppose that $W^*\subseteq \im(A)$. Then
$\dim(A^{-1}(W^*))=\dim(W^*)+\dim(\ker(A))$. Since $W^*= \cB(A^{-1}(W^*))$ and 
$\dim(\ker(A)) = \cork(A)$, it follows that
$A^{-1}(W^*)$ is a $\cork(A)$-singularity witness of $\cB$.

Let us now suppose that some $U \leq V$ is a $\cork(A)$-singularity witness, that is  $\dim(U)-\dim(\cB(U))\geq \cork(A)$.
Then $\dim(U)-\dim(A(U))\geq \cork(A)$ because $A \in \cB$. Since the reverse 
inequality always holds without any condition on $U$, we have 
$\dim(U)-\dim(A(U)) = \cork(A)$. Similarly we have $\dim(U)-\dim(\cB(U)) =  
\cork(A)$, which implies that 
$\dim(A(U)) = \dim(\cB(U))$, and therefore $A(U) = \cB(U)$. For a subspace $S 
\leq V$ the equality $\dim(S)-\dim(A(S)) = \cork(S)$ is equivalent to $\ker(A) 
\subseteq S$, thus we have $\ker(A) \subseteq U$ from which it follows that $U 
= A^{-1}(A(U))$.
But then $\cB^{-1}(A(U)) = \cB^{-1}(\cB(U)) \supseteq U = A^{-1}(A(U))$. Since $W^*$ is the smallest subspace $T \leq V'$ satisfying
$\cB^{-1}(T) \supseteq A^{-1}(T)$, we can conclude that $W^* \subseteq A(U)$.

The existence of  a  $\cork(A)$-singularity
witness obviously implies that $A$ is of maximum rank, and
when $W^*\subseteq \im(A)$ we have already seen that $A^{-1}(W^*)$ is a $\cork(A)$-singularity witness of $\cB$.
Since $\cork(A) = \cork(\cB)$, it is also a $\cork(\cB)$-singularity
witness.
\end{proof}

We remark that in \cite{fr04}, a slightly different version
of this statement is proved. We decided to keep our original
proof for completeness. In our terminology, Theorem~3 of 
\cite{fr04} states that the existence of  a  $\cork(A)$-singularity witness
is equivalent to the equality $\dim(A^{-1}(W^*))=\dim(W^*)+\dim(\ker(A))$.
Both versions offer a straightforward method for testing existence of
(and computing) $\cork(A)$-singularity witnesses.
Besides that our version resembles the concept of augmenting paths 
in algorithms
for matchings in bipartite graphs, it offers the possibility
of stopping the construction of the Wong sequence at the point
after which (while working over the rationals) data blowup can occur; this data 
blowup can occur if we adopt the naive way of computing the preimage of a 
subspace under $A$.
Before that point, we will make use of a pseudo-inverse of $A$.
We describe now this method.


Let $n=\dim(V)$ and $n'=\dim(V')$.
First of all we assume without loss of generality that $n=n'$. 
Indeed, if $n<n'$ we can add as a direct complement
a suitable space to $V$ on which $\cB$ acts as zero, and 
if $n>n'$, we can embed $V'$ into a larger space.
In terms of matrices, this means augmenting the elements
of $\cB$ by zero columns or zero rows to obtain square matrices.
This procedure affects neither the ranks of the matrices in $\cB$
nor the singularity witnesses. 

We say that a nonsingular linear map
$A':V'\rightarrow V$ is a \emph{pseudo-inverse} of $A$ if the restriction of $A'$ to $\im(A)$
is the inverse of the restriction of $A$ to a direct complement
of $\ker(A)$. Such a map can be efficiently constructed as follows. Choose
a direct complement $U$ of $\ker(A)$ in $V$ as well as a direct
complement $U'$ of $\im(A)$ in $V'$. Then take the map
$A_0':\im(A)\rightarrow U$ such that $AA_0'$ is the identity of $\im(A)$
and take an arbitrary nonsingular linear map $A_1':U'\rightarrow \ker(A)$.
Finally let $A'$ be the direct sum of $A_0'$ and $A_1'$. 

\begin{lemma}
\label{lem:second-interpreted}
Let $A\in\cB\leq\matsp$ and let $A'$ be a pseudo-inverse of $A$.
There exists a
$\cork(A)$-singularity 
witness of $\cB$ if and only if 
$(\cB A')^i (\ker(AA')) \subseteq \im(A),$
for all $i \in [n]$.
In the algebraic RAM model as well as over $\Q$,
this can be tested in deterministic polynomial time, and
if the condition holds then $A$ is of maximum rank and $A'(W^*)$ is a  $\cork(\cB)$-singularity
witness which also can be computed deterministically in polynomial time. 
\end{lemma}
\begin{proof}
It follows from Lemma~\ref{lem:second_utility} that a $\cork(A)$-singularity 
witness exists if and only if $W_i \subseteq \im(A)$, for $i = 1, \ldots , n$.
Observing that $(\cB A')^i (\ker(AA')) \subseteq W_i$ for $i = 1, \ldots , n$,
to prove the equivalence it is sufficient to show that if $(\cB A')^i (\ker(AA')) \subseteq \im(A)$
for $i = 1, \ldots , n$ then $W_i = (\cB A')^i (\ker(AA'))$ for $i = 1, \ldots , n$.
The proof is by induction.
For $i=1$ the claim
$W_1= \cB A'(\ker(AA'))$ holds since
$\ker(AA')=A'^{-1} (\ker(A))$.
For $i > 1$, by definition 
$W_{i} = \cB A^{-1} (W_{i-1})$. Since every subspace $W\leq \im(A)$ satisfies $A^{-1}W=A'W + \ker(A)$,
where + denotes the direct sum,
we get $W_{i} \subseteq \cB A' (W_{i-1}) + \cB (\ker(A))$.
Observe that $\cB (\ker(A)) = W_1$. We will show that $W_1 \subseteq \cB A' (W_{i-1})$ 
and then we conclude by the inductive hypothesis. We know that $W_1 \subseteq W_{i-1}$ from the
properties of the Wong sequence, therefore it is sufficient to show that $ W_{i-1} \subseteq \cB A' (W_{i-1})$.
But $W_{i-1} = AA'(W_{i-1})$ since $W_i \subseteq \im(A)$ and $A'$ is the inverse of $A$ on $\im(A)$.

Based on this equivalence, testing the existence 
of a $\cork(A)$-singularity witness can be accomplished
by a simple algorithm. 
First compute a basis
for $\cB$, and then multiply it  by $A'$ to obtain a basis for $\cB A'$. Compute also a basis
for $\ker (AA')$. We now describe how to compute bases for the subspaces in the second Wong sequence
until either we find $i$ such that $W_i \not \subseteq \im(A)$  or we compute $W^*$.
A basis for $W_1$ can be obtained by applying
the basis elements of $\cB A'$ to the basis elements of
$\ker AA'$ and then selecting a maximal set of linearly independent
vectors. Having computed a basis for $W_i$, we stop if
it contains an element outside $\im(A)$. Otherwise
we apply the basis elements of $\cB A'$ to the basis elements
of $W_i$, 
and select a maximal set of linearly independent vectors
to obtain a basis for $W_{i+1}$. When $W_{i+1}=W_i$ we can stop since 
$W^*=W_i$.

If we find that the condition holds then $A'(W^*)$  by Lemma~\ref{lem:second_utility} is a
$\cork(\cB)$-singularity
witness, and it can be easily computed from $W^*$.
\end{proof}

\subsection{The \pep{}}
\label{subsec:pep} 
For $A\in\cB\leq \matsp$,
we would like to know whether $A$ is of maximum
rank in $\cB$. With the help of the limit $W^*$ of the
second Wong sequence of 
$(A, \cB)$ we have established a sufficient condition: we know that
if $W^* \subseteq \im(A)$ then $A$ is indeed of maximum rank. Our results until 
now do not
give a necessary condition for the maximum rank. Now we show that
the second Wong sequence actually allows to translate this question to the \emph{power overflow} problem (PO)
which we define below. As a consequence an efficient solution of the PO
guarantees an efficient solution for the SMR. The reduction is mainly based on 
a theorem of Atkinson and Stephens \cite{AS78} which
essentially says that over large enough fields, in 2-dimensional matrix spaces 
$\cB$,
the equality $\cork(\cB) = \disc(\cB)$ holds.

\begin{fact}[\cite{AS78}]
\label{fact:atkinson}
Assume that $|\F| > n$, and let $A,B \in \matsp$.
If $A$ is a maximum rank element of $\langle A, B\rangle$ 
then there exists a $\cork(A)$-singularity witness of $\langle A, B\rangle$. 
\end{fact}


Combining  Lemma~\ref{lem:second-interpreted}
and Fact~\ref{fact:atkinson}
we get also an equivalent condition for $A$ being of maximum rank.

\begin{lemma}
\label{lem:sufficient}
Assume that $|\F|> n$.
Let $A \in \cB \leq \matsp$, and let $A'$ be a pseudo-inverse of $A$. 
Then $A$ is of maximum rank in $\cB$
if and only if for every $B \in \cB$ and for all $ i\in[n]$, we have 
$$(BA')^i (\ker(AA')) \subseteq \im(A).$$


\end{lemma}

\begin{proof}
First observe that $A$ is of maximum rank in $\cB$
if and only if for every $B \in \cB$, it is of maximum rank in $\langle A, B\rangle$. 
For a fixed $B$, by Fact~\ref{fact:atkinson} and Lemma~\ref{lem:second-interpreted},
$A$ is of maximum rank in $\langle A, B\rangle$ exactly when
$(\langle B, A \rangle A')^i (\ker(AA')) \subseteq \im(A),$ for all $ i\in[n]$.
From that we can conclude since $A'$ is the inverse of $A$ on $\im(A)$.
\end{proof}

This lemma leads us to reduce the problems of
deciding if $A$ is of the maximum rank, and
finding a matrix of rank larger than $A$ when this is not the case,
to the following question.
\begin{problem}[The \pep{}]
Given $\cD\leq\matspc$, $U\leq \F^n$ and $U'\leq \F^n$, 
output $D\in\cD$ and
$\ell\in[n]$ s.t.  $D^\ell (U) \not\subseteq U'$, if there exists such $(D, \ell)$.
Otherwise say {\bf no}.


\end{problem}

The power overflow problem admits an efficient randomized algorithm when $|\F|=\Omega(n)$. For the rank-$1$ spanned case we show a deterministic solution regardless of the field size.

\begin{theorem}
\label{thm:po}
Let $\cD\leq\matspc$ be spanned by rank-1 matrices. Then there exists
$D\in\cD$ and
$\ell\in[n]$ such that  $D^\ell (U) \not\subseteq U'$ if and only if there exists 
$\ell\in[n]$ such that  $\cD^\ell (U) \not\subseteq U'$.
The power overflow problem for $\cD$ can be solved 
deterministically in polynomial time 
on an algebraic RAM as well as over $\Q$.
\end{theorem}

Using this result whose proof 
is given in Section~\ref{sec:rank1} we are now ready to prove 
Theorem~\ref{thm:main1}.
\begin{proof}[\bf Proof of Theorem~\ref{thm:main1}.]
First we suppose that $|\F| \geq n+1$.
Let $A$ be an arbitrary matrix in $\cB$. The algorithm iterates the following process until
$A$ becomes of maximum rank.

We run the algorithm of Lemma~\ref{lem:second-interpreted} to test whether
$(\cB A')^i (\ker(AA')) \subseteq \im(A)$
for $i \in [n]$.
If this condition holds then $A$ is of maximum rank,
and the algorithm also gives a $\cork(\cB)$-singularity witness.
Otherwise we know by Theorem~\ref{thm:po} that there exists $B \in \cB$ and $i \in [n]$ such that
$(BA')^i (\ker(AA')) \not \subseteq \im(A)$.
We
apply the algorithm of Theorem~\ref{thm:po}
with input $\cB A'$, $\ker (AA')$ and $\im (A)$,
which finds such a couple $(B,i)$.
Lemma~\ref{lem:sufficient} applied to  $\langle A, B\rangle$ implies that
$A$ is not of maximum rank in $\langle A, B\rangle$.
If $A$ has rank $r \leq n-1$ which is not maximal in
$\langle A, B\rangle$, then the determinant of an 
appropriate $(r+1) \times (r+1)$ minor of $A+\lambda B$
is
a nonzero polynomial of degree at most $r+1$
which has at most $r+1 \leq n$ roots.
We then pick $n+1$ arbitrary field elements $\lambda_1, \ldots, \lambda_{n+1}$, and 
we know that for some $1 \leq j \leq n+1$ we have $\rk(A+\lambda_j B) > \rk(A)$. We replace
$A$ by $A+\lambda_j B$ and restart the process.

Over $\Q$, 
at the end of each iteration, by a reduction procedure described in \cite{GIR}  
we can achieve that the matrix $A$,
written as a linear combination of $B_1,\ldots,B_m$ has
coefficients from a fixed subset $K \subseteq \Q$ of size $n+1$
(say, $K=\{0,\ldots,n\}$).
In fact, if 
$A=\alpha_1B_1+\alpha_2B_2\ldots+\alpha_mB_m$ has rank $r$ then
for at least one $\kappa_1\in K$ the matrix
$\kappa_1B_1+\alpha_2B_2\ldots+\alpha_mB_m$ has rank at least $r$.
This way all the coefficients $\alpha_j$ can be replaced
with an appropriate element from $K$.


As in each iteration we either stop (and conclude with $A$ being of maximal 
rank), or increase the rank of $A$ by at least $1$, the number of iterations is 
at most $n$. Also, each iteration takes polynomial many steps since the 
processes 
of Lemma~\ref{lem:second-interpreted} and 
Theorem~\ref{thm:po} are polynomial.
Therefore the overall running  time is also polynomial. This finishes the case 
for $|\F|\geq n+1$.

When $|\F|<n+1$, 
we can compute the maximum rank by running the above procedure over a
sufficiently large extension field. The maximum rank
will not grow if we go over an extension. This follows
from the fact that the equality $\cork(\cB)=\disc(\cB)$ holds
over any field if $\cB$ is spanned by an arbitrary 
matrix and by rank one matrices, see \cite{IKS}.
\end{proof}
\begin{remark}\label{remark:thm1}
As mentioned in the introduction, we can generalize to the setting when 
$\cB$ is spanned by rank-$1$ matrices and an arbitrary matrix, as follows. Let 
$\cB'$ be the subspace of $\cB$ generated by rank-$1$ matrices. 
As indicated in Lemma~\ref{lem:goodproduct1} in the next section, in this case 
the algorithm for \pep{} with special $U=\ker(AA')$ and $U'=\im(A)$ is still
guaranteed to succeed if $A\not\in \cB'$. Furthermore, in the update step, the 
resulting matrix of higher rank keeps the property of not in $\cB'$. So from 
the given basis $B_1,\ldots,B_m$ for $\cB$, we apply the procedure in 
Theorem~\ref{thm:main1} 
using $B_i$ as the starting point, for each $B_i$. Then it is ensured that for 
those $B_i\not\in\cB'$ this procedure will succeed in finding a matrix with 
maximal rank. Otherwise if $B_i\in \cB'$, then \pep{} will either get 
$\mathbf{no}$, or detect that $\cH_\ell\cdots \cH_1(U)\subseteq U'$, so ending 
with a safe return. 
\end{remark}

\section{The \pep{} for rank-$1$ spanned matrix spaces}\label{sec:rank1}

In this section we prove {\bf Theorem~\ref{thm:po}}.

\paragraph{The setting.} Given subspaces $U,U'$ of $\F^n$ as well
as a basis $\{D_1, \dots, D_m\}$ for a matrix space
$\cD\leq \matspc$, we will show is that in polynomial time we can decide if $\cD^\ell (U) \not\subseteq U'$
for some $\ell$, and if this holds then find $D \in \cD$ s.t. $D^\ell (U) \not\subseteq U'$.


Formally let $\ell =\ell(\cD)$ be the smallest integer
$j$ s.t. $\cD^j (U) \not \subseteq U'$
 if such an integer exists, and $n$ otherwise.
We start by computing $\ell$ and for $1 \leq j \leq \ell$, bases $\cT_j$ for $\cD^j$.
Set  $\cT_1 =$ $ \{D_1,\ldots,D_m\}$. 
If $\cD^j (U) \not \subseteq U'$ then we set $\ell =j$ and stop constructing further bases.
If $j = n$ and $\cD^n (U)  \subseteq U'$ then we stop the algorithm  and output {\bf no}.
Otherwise
we compute $\cT_{j+1}$ by selecting
a maximal linearly independent set form the products of elements in $\cT_{j}$ and $\cT_{1}$.


\paragraph{Helpful subspaces of $\cD$.} Recall that our goal is to find $D$ 
such that $D^\ell (U) \not\subseteq U'$. To achieve this goal, 
for $i\in[\ell]$, 
we define subspaces $\cH_i$ of $\cD$, 
which play a crucial role in the algorithm:
$$
\cH_i=\{X\in\cD\mid  \cD^{\ell-j}X \cD^{j-1} (U) \subseteq U',~
j=1,\ldots,i-1,i+1,\ldots,\ell\}.
$$
Let us examine the meaning for some matrix $X$ to be in $\cH_i$. Let $P$ be a 
product of $\ell$ elements from $\cD$, and suppose $X$ appears in $P$. Then 
$X\in\cH_i$ implies that, as long as $X$ appears in $P$ at the $j$th position, 
$j\neq i$, then it must be that $P (U) \subseteq U'$. In other words, for $P$ 
to be able to pull $U$ out of $U'$, it is necessary that $X$ appears at the 
$i$th position. 

The following lemma explains why $\cH_i$'s are useful for the purpose of 
powerflow problem. 
\begin{lemma}
\label{lem:oneterm}
For a matrix $X=X_1+\ldots+X_\ell$ with
$X_i\in \cH_i$, we have
$X^\ell (U) \subseteq U'$ if and
only if $X_\ell\cdots X_2X_1(U) \subseteq U'$.
\end{lemma}
\begin{proof}
We have $X^\ell=\sum_{\sigma}X_{\sigma(\ell)}\cdots 
X_{\sigma(1)}$,
where
the summation is over the maps $\sigma:[\ell]\rightarrow [\ell]$.
When $\sigma$ is not the identity map then there exists an index $j$
such that $\sigma(j)\neq j$. 
Then
$X_{\sigma(\ell)}\cdots X_{\sigma(1)}(U) \subseteq U'$ by the definition
of $\cH_{\sigma(j)}$.
\end{proof}

Furthermore, $\cH_i$'s can be computed efficiently 
on an algebraic RAM as well as over $\Q$
as follows. 
Let $x_1, \dots, x_m$ be formal variables,
an element in $\cD$ can be written as 
$X=\sum_{k\in[m]}x_kD_k$. The condition 
$\cD^{\ell-j}X \cD^{j-1} (U) \subseteq U'$
is equivalent to the set of the following 
homogeneous linear equations in the variables $x_k$: 
$\langle Z( \sum_{k\in[m]}x_kD_k)Z'u,v\rangle= 0 ,$ 
where $Z$ is from $\cT_{\ell-j}$, $Z'$ is from
$\cT_{j-1}$, $u$ is from a basis for $U$ and
$v$ is from a basis for ${U'}^\perp$. Thus
$\cH_i$ can be computed by solving
a system of polynomially many homogeneous
linear equations. Note that the coefficients of 
the equations are scalar products of vectors from a basis
for $U'^\perp$ by vectors obtained as applying
products of $\ell$ matrices from $\{D_1,\ldots,D_m\}$
to basis elements for $U$.

\paragraph{Back to rank-$1$ spanned setting.} 
In general, $\cH_i$ can be $0$. 
In our setting, due to the existence of a basis of rank-$1$ matrices, 
fortunately this is far from the case. 

\begin{lemma}
\label{lem:goodproduct}
Suppose $\cB$ is rank-$1$ spanned, and $\ell$ is the smallest integer such that 
$\cD^\ell (U) \not\subseteq U'$. Then the following hold:
(1) $\forall i\in[\ell]$, $\cH_i\neq 0$;
(2) $\cH_\ell \cdots \cH_1(U) \not\subseteq U'$.
\end{lemma}

\begin{proof}
Assume that $\cD$ is spanned by 
the rank one matrices $C_1,\ldots,C_m$, where $C_i$ may be over an extension 
field $\F'$ of $\F$. 

Let us first consider the case when $C_i$'s are matrices over $\F$. Then there
exist indices $k_1,\ldots,k_{\ell}$ such
$C_{k_\ell}\cdots C_{k_1} (U) \not\subseteq U'.$
We show that $C_{k_i}\in \cH_i$, for $i \in [\ell]$, proving (1). This also 
implies immediately
$\cH_\ell \cdots \cH_1 (U) \not\subseteq U'$, proving (2). 

Assume by contradiction  that $C_{k_i}\not\in \cH_i$, for some $i \in [\ell]$. Then
$\cD^{\ell-j} C_{k_i} \cD^{j-1} (U) \not\subseteq U',$
for some $j\neq i$. On the other hand $C_{k_i}$ satisfies 
$\cD^{\ell-i} C_{k_i} \cD^{i-1} (U) \not\subseteq U'.$
Since $C_{k_i}$ is of rank $1$ we have  
$C_{k_i} \cD^{j-1} (U) =C_{k_i} \cD^{i-1} (U),$
which yields that neither 
$\cD^{\ell-i} C_{k_i}\cD^{j-1} (U)$
nor 
$\cD^{\ell-j} C_{k_i}\cD^{i-1} (U)$
is contained in $U'$.
However one of these products
 is shorter than $\ell$, contradicting the minimality of $\ell$. 
 
To generalize to $C_i$'s over an extension field $\F'$, it suffices to lift all 
objects ($\cD$, $\cH_i$, $U$ and $U'$) to their spans with the extension field 
$\F'$ (denoted as $\F'\cD$, $\F'\cH_i$, $\F U$ and $\F U'$). After going 
through the above argument, we have $\F'\cH_i\neq 0$ and $\F'\cH_\ell\cdots 
\F'\cH_1(\F'U)\not\subseteq \F'U'$. We then have $\cH_i\neq 0$ and 
$\cH_\ell\cdots \cH_1(U)\not\subseteq U'$, as it is not hard to see that 
$\F'\cH_i$, and $\F'\cH_\ell\cdots \F'\cH_1(\F'U)$, are spans of $\cH_i$ and 
$\cH_\ell\cdots \cH_1(U)$ with the extension field $\F'$. 
\end{proof}
That is, in out setting, not only $\cH_i\neq 0$, but $\cH_\ell\cdots \cH_1$ is 
able to pull $U$ outside $U'$ (instead of using the full power of $\cD^\ell$).

To finish the algorithm, we compute bases for products $\cH_i\cdots\cH_1$, for 
$i \in [n]$,
in a way similar to computing bases for $\cD^i$. Then we search the basis
of $\cH_{\ell}$ for an element $Z$ such that
$Z\cH_{\ell-1}\cdots\cH_1(U) \not\subseteq U'$. 
We put $X_{\ell}=Z$ and continue
searching the basis of $\cH_{\ell-1}$ for an element $Z$ 
such that 
$X_\ell Z\cH_{\ell-2}\cdots\cH_1 (U) \not\subseteq U'$.
Continuing the iteration, Lemma~\ref{lem:goodproduct} ensures that eventually 
we find 
$X_i\in \cH_i$,  for $i\in[\ell]$, such that $X_\ell\cdots X_1(U) \not\subseteq 
U'$.
We set $D=X_1+\ldots+X_\ell$, then by Lemma~\ref{lem:oneterm} we have
$D^\ell (U) \not\subseteq U'$.  We return $D$ and $\ell$. This finishes the 
proof of Theorem~\ref{thm:po}. 

Finally, we introduce the following slight extension of 
Lemma~\ref{lem:goodproduct} for special subspaces 
$U,U'$, as applicable to Remark~\ref{remark:thm1} (2). 
\begin{lemma}
\label{lem:goodproduct1}
Assume that $\cD$ is spanned by rank one matrices
and a projection to $U'$ having kernel $U$.
Then $\cH_\ell \cdots \cH_1U\not\subseteq U'$.
\end{lemma}

\begin{proof}
Identical with the proof of Lemma~\ref{lem:goodproduct},
based on the observation that a projection 
with the prescribed properties can be deleted from 
any product mapping $U$ outside $U'$. 
\end{proof}



\section{The first Wong sequence and triangularizable matrix spaces}\label{sec:first}
\subsection{The connection}

 
To tackle the triangularizable matrix spaces, our starting point is the 
following lemma, which connects first Wong sequences 
with singularity witnesses. 
\begin{lemma}\label{prop:first_1_dim}
Let $A\in\cB\leq\matspc$, and let $U^*$ be the limit of the first Wong sequence of $(A, \cB)$.
Set $d=\dim(U^*)$. Then either $U^*$ is a singularity witness of $\cB$, or there exist nonsingular 
matrices $P, Q\in\matspc$, such that $\forall B\in \cB$, $QBP^{-1}$ is of the 
form {\small
$\left[
\begin{array}{cc}
X & Y \\
0 & Z
\end{array}
\right]$
}, where $X$ is of size $d\times d$, and $\cB$ is nonsingular in the $X$-block.
\end{lemma}
\begin{proof}
If $\dim(U^*)>\dim(\cB(U^*))$ then $U^*$ is a singularity witness. If 
$\dim(U^*) = \dim(\cB(U^*))$ then the choice of $P$ and $Q$
corresponds to an appropriate basis change transformation.
To see that $\cB$ is nonsingular in the $X$-block, note that $A \in \cB$ and 
$A(U^*) = \cB (U^*)$.
\end{proof}

Lemma~\ref{prop:first_1_dim} suggests a recursive algorithm: take an arbitrary $A\in\cB$ and compute $U^*$, the limit of the first Wong sequence of $(A, \cB)$. If we get a
singularity witness, we are done. Otherwise, if $U^*\neq 0$, as the $X$-block is already nonsingular, we only need to focus on the nonsingularity of $Z$-block which is of smaller size. To make this idea work, we have to satisfy essentially two
conditions. We must find some $A$ such that $U^*\neq 0$, and to allow for recursion the specific property of the matrix space $\cB$ 
we are concerned with has to
be inherited by the subspace corresponding to the $Z$-block. It turns out that 
in the triangularizable case these two problems can be 
taken care of by the following lemma.




\begin{lemma}\label{lem:first_correct}
Let $\cB\leq \F$ be given by a basis $\{B_1, \dots, B_m\}$, and suppose that there exist nonsingular matrices
$C, D\in M(n, \F')$ such that $B_i=DB_i'C^{-1}$, and $B_i'\in M(n, \F')$ is 
upper triangular for every
$i\in[m]$. 
Then we have the following. 
\begin{enumerate}
\item Either $\cap_{i\in[m]}\ker(B_i)\neq 0$, or $\exists j\in [m]$ and $0\neq 
U\leq \F^n$ s.t. $B_j(U)=\cB(U)$. 
\item Suppose there exist $j\in [m]$ and $0\neq U\leq \F^n$ s.t. $B_j(U)=\cB(U)$, and $\dim(U)=\dim(B_j(U))$. 
Let $B_i^* : \vecspc/U \rightarrow \vecspc/\cB(U)$ 
be the linear map induced by $B_i$, for $i\in[m]$.
Then $\cB^*=\langle B_1^*, \dots, B_m^*\rangle$ 
is triangularizable over $\F'$. 
\end{enumerate}
\end{lemma}
\begin{proof}
1. Let $\{e_i\mid i\in[n]\}$ be the standard basis of ${\F'}^n$, 
and $c_i=C(e_i)$ and $d_i=D(e_i)$ for $i\in[n]$. 
If  $B_i'(1, 1)=0$ for all $ i\in[m]$ then $c_1$ is in the 
kernel of every $B_i$'s. 
If there exists $j$ such that $B_j'(1, 1)\neq 0$, we
set $U'=\langle c_1\rangle\leq \F'^n$. 
 Then it is clear that 
$\langle d_1\rangle = B_j(U')=\cB(U')$.
It follows that 
the first Wong sequence of $(B_j, \cB)$ over $\F'$ has nonzero limit,
and therefore the same holds over $\F$. We can choose for $U$ this limit.


2. 
First we recall that for a vector space $V$ of dimension $n$, 
a complete flag of $V$ is a nested sequence of subspaces 
$0=V_0\subset V_1 \subset\dots \subset V_n=V$. 
For $\cA\leq \matsp$ with $\dim(V)=\dim(V')=n$, the matrix space
$\cA$ is triangularizable if and only if $\exists$ complete 
flags $0=V_0\subset V_1 \subset\dots \subset V_n=V$ 
and $0=V_0'\subset V_1' \subset\dots \subset V_n'=V'$ s.t.
$\cA(V_i)\subseteq V_i'$ for $i \in [n]$. 

For $U\leq \F^n$,
let $\F'U$ be the linear span of $U$ in ${\F'}^n$. We think of $B_i$'s and $B_i^*$'s as linear maps over $\F'$ in a natural way. 
Let $\ell=\dim({\F'}^n/\F'U)$. For $0 \leq i \leq n$ set $S_i=\langle c_1, \dots, c_i\rangle$ and $T_i=\langle d_1, \dots, d_i\rangle$.
Obviously $\cB(S_i)\subseteq T_i$ for $0 \leq i \leq n$. Let $S_i^*=S_i/\F'U$ and $T_i^*=T_i/\cB(\F'U)$, and 
consider $S_0^*\subseteq \dots\subseteq S_n^*$ and $T_0^*\subseteq \dots \subseteq T_n^*$. 
We claim that $
\forall i\in [n], \dim(S_i^*)\geq \dim(T_i^*).$
This is because as $T_i\cap \cB(\F'U)\supseteq B_j(S_i\cap \F'U)$, by $\dim(\F'U)=\dim(B_j(\F'U))$, $\dim(B_j(S_i\cap \F'U))\geq \dim(S_i\cap \F'U)$. Thus $\dim(S_i\cap \F'U)\leq \dim(T_i\cap \cB(\F'U))$, and $\dim(S_i^*)\geq \dim(T_i^*)$. 
As $\cB^*(S_i^*)\subseteq T_i^*$, $\dim(S_{i+1}^*)-\dim(S_i^*)\leq 1$, 
and $\dim(T_{i+1}^*)-\dim(T_i^*)\leq 1$, there exist two nested sequences $S_0^*\subset S_{j_1}^*\subset\dots \subset S_{j_\ell}^*=S_n^*$ 
and $T_0^*\subset T_{k_1}^*\subset\dots\subset T_{k_\ell}^*=T_n^*$, s.t. $\dim(S_{j_h})=\dim(T_{k_h})=h$. Furthermore, by $\dim(S_i^*)\geq \dim(T_i^*)$, $j_h\leq k_h$, thus $\cB^*(S^*_{j_h})\subseteq \cB^*(S^*_{k_h})\subseteq T^*_{k_h}$, $\forall h\in[\ell]$. That is, the two nested sequences are complete flags, and $\cB^*$ is triangularizable over $\F'$.
\end{proof}

\subsection{An algorithm 
on an algebraic RAM}


Suppose we are given a basis $\{B_1, \dots, B_m\}$ for 
$\cB\leq \matspc$ which is triangularizable over
an extension field $\F'$ of $\F$, i.~e., $B_i=DB_i'C^{-1}$ for some nonsingular 
$C, D\in\matspe$, and $B_i'\in \matspe$ is upper triangular for every
$i\in[m]$. 
Our problem is to determine whether 
there exists a nonsingular matrix in $\cB$ or not and 
finding such a matrix
if exists.

Given the preparation of Lemma~\ref{lem:first_correct}, here is the outline of 
an algorithm using 
polynomially many arithmetic operations. The algorithm recurses on the size of 
the matrices, with the base case being the size $1$. 
It checks at the beginning whether $\cap_{i\in[m]}\ker(B_i) = 0$. 
If this is the case then it returns $\cap_{i\in[m]}\ker(B_i)$ which is a 
singularity witness.
Otherwise, for all $i\in[m]$,  it computes the limit $U_i^*$ of the 
first Wong sequence for $(B_i, \cB)$. By Lemma~\ref{lem:first_correct} 
(1) there exists $j\in[m]$ such that $U_j^* \neq 0$ and $B_j(U)=\cB(U)$. 
The algorithm then recurses on the induced actions $B_i^*$'s of $B_i$'s, 
which are also triangularizable by Lemma~\ref{lem:first_correct} (2). 
When $\cB$ is nonsingular the algorithm should return 
a nonsingular matrix. This nonsingular matrix is built step by step by the 
recursive calls,
at each step  we have to construct a nonsingular  linear 
combination of $B_j$ and the matrix returned by the call.
For this we need $n+1$ field elements. 

We expand the above idea into a rigorous algorithm, called $\trialgo$ and 
present it in Algorithm~\ref{algo:tri}. This algorithm requires polynomially 
many arithmetic operations, and therefore of polynomial complexity in finite 
fields. 
The input of the algorithm can be an arbitrary matrix space (not necessarily 
triangularizable), but it may fail in certain cases. For triangularizable 
matrix spaces the algorithm would not fail due to 
Lemma~\ref{lem:first_correct}. 
Note that though the algorithm works assuming triangularizability over some 
extension field, the algorithm itself does not need to deal with the field 
extension explicitly by Lemma~\ref{lem:first_correct}, given that $\F$ is large 
enough. 
To allow for recursion, the output of the algorithm can be one of the 
following: the first is an explicit linear combination of the given matrices, 
which gives a nonsingular matrix. 
The second is a singular subspace witness. The third one is $\fail$.
\begin{algorithm}
  \caption{\trialgo$(B_1, \dots, B_m)$}\label{algo:tri}
  \KwIn{$\cB=\langle B_1, \dots, B_m\rangle\subseteq \matspc$.}
  \KwOut{One of the following: (1) $(\alpha_1, \dots, \alpha_m)\in\F^m$ s.t. 
  $\sum_{i\in[m]}\alpha_i B_i$ is nonsingular. (2) A singular subspace witness 
  $U\leq \vecspc$. (3) $\fail$.}
  
  \tcp{Base case} 
  \If{$n=1$}
  {If $\exists$ nonzero $B_i$, \Return $(0, \dots, 1, \dots, 0)$ where $1$ is 
  at the $i$th position. Otherwise \Return $\vecspc$. }
  \tcp{Start of the recursive step.}
  \tcp{If $\cap_{i\in[m]}\ker(B_i)\neq 0$ then $\cap_{i\in[m]}\ker(B_i)$ is a 
  singular witness itself.}
  \If{$\cap_{i\in[m]}\ker(B_i)\neq 0$}{\Return $\cap_{i\in[m]}\ker(B_i)$.}
  \ForAll{$i\in[m]$}{$U_i^*\gets$ the limit of the first Wong sequence of 
  $(B_i, \cB)$.}
  \If{$\exists i\in[m]$, $\dim(\cB(U_i^*))<\dim(U_i^*)$}{\Return $U_i^*$}
  \If{$\not\exists j$ s.t. $\dim(U_j^*)>0$}{\Return $\fail$}
  $U^*\gets U_j^*$ where $U_j^*$ satisfies that $\dim(U_j^*)>0$\;
  \ForAll{$i\in[m]$}{$B_i^*\gets$ the induced linear map of $B_i$ from 
  $\vecspc/U^*$ to $\vecspc/\cB(U^*)$.}
  \tcp{Recursive call.}
  $X\gets\trialgo(B_1^*, \dots, B_m^*)$\;
  \If{$X$ is a singular subspace witness $W/U^*$}{\Return the full preimage of 
  $W/U^*$ in the canonical projection $\vecspc\to \vecspc/U^*$.}
  \ElseIf{$X$ is $(\alpha_1, \dots, \alpha_m)$}
  {
  $\Lambda\gets$ a set of field element of size $n+1$\;
  $E\gets\sum_{i\in[m]}\alpha_i B_i$\;
  Choose $(\lambda, \mu)\in\Lambda\times\Lambda$, s.t. $\lambda B_j+\mu E$ is 
  nonsingular\;
  \Return $(\mu\alpha_1, \dots, \mu\alpha_{j-1}, \mu\alpha_j+\lambda, 
  \mu\alpha_{j+1}, \dots, \mu\alpha_m)$
  }
  \ElseIf{$X=\fail$}
  {\Return $\fail$}
\end{algorithm}

Regarding implementation, it might be needed to comment on Line 13. At this 
point we have that $\dim(U^*)\leq \dim(\cB(U^*))\leq \dim(B_j(U^*))\leq 
\dim(U^*)$. Thus $\dim(\cB(U^*))=\dim(U^*)$, and note that $B_i(U^*)\subseteq 
\cB(U^*)$, for all $i\in[m]$. 
Then two bases of $\vecspc$ can be formed by extending bases of $U^*$ and 
$\cB(U^*)$ respectively, and w.r.t. these two bases the induced action $B_i$ 
from $\vecspc/U^*$ to $\vecspc/\cB(U^*)$ can be read off easily. 

For correctness we distinguish among the types of output of the algorithm, and 
show that they indeed have the required property. 
\begin{description}
\rmitem[If $(\alpha_1, \dots, \alpha_m)$ is returned: ] This case occurs in 
Line 2 and Line 20. Line 2 is trivial. If the algorithm reaches Line 20, we 
claim that there exists $(\lambda, \mu)\in\Lambda\times\Lambda$ s.t. $\lambda 
B_j+\mu E$ is nonsingular. Let $P$ and $Q$ be the matrices from 
Lemma~\ref{prop:first_1_dim}. Thus $\forall i\in[m]$, $QB_iP^{-1}$ is of the 
form:
$
\left[
\begin{array}{cc}
X_i & Y_i \\
0 & Z_i
\end{array}
\right],
$
where $X_i$ is of size $(n-\ell)\times (n-\ell)$ and $Z_i$ is of size $\ell$ by 
$\ell$. As $X_j$ is nonsingular and $\sum_{i\in[m]}\alpha_iZ_i$ is nonsingular, 
$\det(xB_j+yE)$ is a nonzero polynomial, thus from Schwartz-Zippel lemma the 
existence of $(\lambda, \mu)$ in $\Lambda\times\Lambda$ is ensured.
\rmitem[If a subspace of $\vecspc$ is returned: ] This case occurs in Line 2, 
4, 8 and 16. All are straightforward.
\rmitem[The case of $\fail$: ] After Line 3 $\cap_{i\in[m]}\ker(B_i)= 0$. Then 
Lemma~\ref{lem:first_correct} ensures that $\fail$ cannot be returned for 
triangularizable matrix spaces. 
\end{description}

\subsection{An algorithm over the rationals}

To obtain a polynomial-time algorithm over rationals, we give first a  
characterization of  triangularizability 
of a nonsingular matrix space.

\begin{lemma}
\label{lem:triang-reg}
Assume $\cB\leq \matspc$ contains a nonsingular matrix $S$.
Then $\cB$ is triangularizable over $\F$ if and only if
there exists a nonsingular matrix $D\in \matspc$ such
that $D^{-1}\cB S^{-1}D$ consists of upper triangular matrices.
\end{lemma}

\begin{proof}
$\Rightarrow$: 
Assume that $D^{-1}\cB C$ consists of upper triangular matrices.
Then $C^{-1}S^{-1}D=(D^{-1}S C)^{-1}$ is upper triangular as well,
whence -- as products of upper triangular matrices remain upper triangular --
$D^{-1}\cB S^{-1}D=(D^{-1}\cB C)(C^{-1}S^{-1}D)$ also
consists of upper triangular matrices.
\\
$\Leftarrow$: Assume that $D^{-1}\cB S^{-1}D$ consists of upper triangular
matrices. Put $C=S^{-1}D$.
\end{proof}

We have the following criterion of triangularizability:

\begin{lemma}\label{lem:tri_nonsing}
Let $\cA\leq \matspc$ containing the identity matrix
and let $\F'$ be the algebraic closure of $\F$. 
Then there exists $D\in \matspe$ such that $D^{-1}\cA D$ consists
of upper triangular matrices (over $\F'$) if and only if
$$(\cA^{n^2}[\cA,\cA]\cA^{n^2})^n=(0).$$
Here $[\cA,\cA]$ is the space spanned by 
the commutators $[X,Y]=XY-YX$ ($X,Y\in\cA$).
\end{lemma}

\begin{proof}
Put $\cD=\cA^{n^2}$. 
Then $\cD$ is the matrix algebra
generated by $\cA$. The formula expresses that the two-sided
ideal of $\cD$ generated by the commutators from $\cA$ 
is nilpotent. Let $\cD'=\F'\otimes \cD$. Then the formula
is also equivalent to that the ideal of $\cD'$ generated
by the commutators is nilpotent. This is further equivalent
to that the factor algebra $\cD'/Rad(\cD')$ is commutative.
However, over an algebraically closed field a matrix
algebra is nilpotent if it is a conjugate of a
subalgebra of the upper triangular matrices.
(To see one direction, observe that the whole algebra of the 
upper triangular matrices and hence every subalgebra 
of it has this property. As for reverse implication, note that all 
the irreducible representations of an algebra over an algebraically
closed field which is commutative by its radical are one-dimensional 
and hence a composition series
gives a complete flag consisting of invariant subspaces.)
\end{proof}

\begin{corollary}\label{cor:tri_nonsing}
Assume that we are given a nonsingular $S\in\cB$. Then
there is a polynomial time algorithm (on an algebraic RAM
as well as the case $\F=\Q$)
which decides
whether or not there exists an extension of $\F$ over
which $\cB$ is triangularizable.
\end{corollary}

Again, we actually have an algorithm using
a polynomial number of arithmetic operations
and equality tests in the black box model for $\F$.

With these preparations we are now ready to prove Theorem~\ref{thm:main2}.\\ \\
{\bf Proof of Theorem~\ref{thm:main2}.}
On an algebraic RAM
Algorithm~\ref{algo:tri} is all we need. 
Over rationals we shall perform a reduction to finite fields via 
Lemma~\ref{lem:tri_nonsing}.

We assume that $\cB$ is given by matrices $B_1,\ldots,B_m$
over $\Q$. Multiplying by a common denominator for the
entries, we can achieve the situation when the entries
of $B_1,\ldots,B_m$ are integers. Let $b$ be a bound on the
 of absolute values of the entries of $B_1,\ldots,B_m$. Then
a polynomial-time algorithm should run in time polynomial in
$n$ and $\log b$. If $\cal B$ is nonsingular then there exist integers
$\lambda_1,\ldots,\lambda_m$, each between $0$ and $n$
such that $S=\lambda_1B_1+\ldots+\lambda_mB_m$ is
nonsingular. The absolute value of the determinant of $S$ is 
a nonzero integer whose logarithm is bounded by a polynomial
in $\log b$ and $n$. It follows that there is a prime 
$p$ bounded by an (explicit) polynomial in $\log b$ and
$n$ that does not divide the determinant of $S$. 

Let $S'=\det(S)S^{-1}$. 
We reduce the problem modulo $p$. We see that $S$ and $S'$
are integral matrices and both are invertible module $p$. 
Furthermore, if $\cB$ is triangularizable over an extension
of $\Q$, by Lemma~\ref{lem:tri_nonsing} all length-$n$ products of elements of 
the form
$B_{i_1}S'\cdots B_{i_{n^2}}S'
[B_{j_1}S',B_{j_2}S'] B_{k_1}S'\cdots B_{k_{n^2}}S'$
vanish, and this will be the case modulo $p$ as well.
It follows that the subspace of matrices over $\F_p$,
spanned by the matrices  $B_iS'$, reduced modulo $p$,
can be triangularized over an extension field of $\F_p$.
But then the space spanned by $B_i$ is also
triangularizable
(over the same extension).
  
Thus if $\cB$ is nonsingular and triangularizable over 
an extension of $\Q$ then there is a prime $p$ greater than
$n$ but smaller than the value of an explicit polynomial function
in $\log b$ and $n$, such that the reduction modulo $p$ gives a nonsingular
matrix space which is triangularizable over an extension
field of $\F_p$. The algorithm consists of taking the primes
$p$ up to the polynomial limit and applying the generic method 
over $\F_p$ to the reduced setting. The method either finds a $p$
and an integer combination of $B_1,\ldots,B_m$ which
is nonsingular even modulo $p$ or, concludes
that $\cB$ cannot be nonsingular and triangularizable at the same time. $ 
\hfill \Box$

\section{On the Edmonds-Rado class and some subclasses}\label{sec:more_ER}

\subsection{Matrix spaces not in the Edmonds-Rado class}

Recall that $\cB\leq\matspc$ is in the Edmonds-Rado class if either $\cB$ 
contains nonsingular matrices, or $\cB$ is singular and there exists a 
singularity witness of $\cB$. Recall that in Section~\ref{sec:intro} we 
defined the discrepancy
$\disc(\cB)=\max\{c\in \N\mid \exists~ c$-singularity witness of $\cB \}$, and 
from the definition it is clear 
that $\mincork(\cB)\geq \disc(\cB)$. In terms of discrepancy,  
$\cB$ is in the Edmonds-Rado class, if $\disc(\cB)=0 \iff \cork(\cB)=0$. 

A well-known example of a matrix space not in the Edmonds-Rado class is the 
class $\sk_3$ of 3 dimensional skew symmetric
matrices, generated for example by the following:
$$
\sk_3=
\left\langle 
\left[\begin{array}{ccc}
0 & 1 & 0\\
-1 & 0 & 0\\
0 & 0 & 0
\end{array}\right], 
\left[\begin{array}{ccc}
0 & 0 & 0\\
0 & 0 & 1\\
0 & -1 & 0
\end{array}\right],
\left[\begin{array}{ccc}
0 & 0 & 1\\
0 & 0 & 0\\
-1 & 0 & 0
\end{array}\right]
\right\rangle. 
$$

Consider the following problem: if a matrix space $\cB$ has a basis consisting 
of matrices with certain property,
does it imply that $\cB$ is in the Edmonds-Rado class? 
Gurvits has observed that if $\cB$ has a basis consisting of triangular or 
semidefinite matrices then it is in
the Edmonds-Rado class. 
We now show that the other two basis properties, namely consisting of 
projections or positive matrices, do not necessarily 
imply that $\cB$ is in the Edmonds-Rado class. Recall that a matrix over $\R$ 
is positive if every entry in it is positive.


Let $\cB=\langle B_1, \dots, B_m\rangle\leq M(n, \F)$, and let $A \in M(n, \F)$ 
be an arbitrary nonsingular matrix. 
For $i \in [m],$ we define
$
Y_i=\left[
\begin{array}{cc}
A & B_i \\
0 & 0
\end{array}
\right]$ and 
$
Z=\left[
\begin{array}{cc}
0 & 0\\
A & 0
\end{array}
\right]
$, and 
let $\cA=\langle Y_1, \dots, Y_m, Z\rangle$.

\begin{lemma}
\label{lem:whatever}
We have $\disc(\cA)=\disc(\cB)$.
\end{lemma}
\begin{proof}
Let $E_1\leq \F^{2n}$ be the coordinate subspace generated by the first $n$ 
coordinates, and $E_2\leq \F^{2n}$ be the coordinate subspace generated by the 
last $n$ coordinates. 

To show that $\disc(\cA)\geq \disc(\cB)$, we take a $\disc(\cB)$-discrepancy 
witness $U$ of $\cB$, 
and embed $U$ into $E_2$. Then $\langle E_1\cup U\rangle$ 
is a $\disc(\cB)$-singularity witness of $\cA$. 

To show that $\disc(\cA)\leq \disc(\cB)$, let $W$ be $\disc(\cA)$-singularity 
witness of $\cA$. 
Let $W_1'=W'\cap E_1$ and $W_2'=W'\cap E_2$. Due to the form of 
the $Y_i$'s and $Z$, $W'=W_1'\oplus W_2'$. In particular note that $W_2'=Z(W)$, 
and $A$ is nonsingular. So if we set  
$R = \{w\in W\mid Z(w)=0\}$, 
we have $R\leq E_2$, $\dim(R)=\dim(W)-\dim(W_2')$, and 
$\dim(W_1')\geq \dim(\cB(R))$. 
Thus $\disc(\cA)=\dim(W)-\dim(W')=(\dim(W)-\dim(W_2'))-\dim(W_1') \leq 
\dim(R)-\dim(\cB(R))\leq \disc(\cB)$. 
\end{proof}

\begin{proposition}
There exist matrix spaces generated by projections or positive matrices 
outside the Edmonds-Rado class
\end{proposition}
\begin{proof}
For $i \in [m],$ we define $Y_i', Z'\in M(2n, \F)$ by
$
Y'_i=\left[
\begin{array}{cc}
A & B_i+A \\
0 & 0
\end{array}
\right]
$
and 
$
Z'=\left[
\begin{array}{cc}
0 & 0\\
A & A
\end{array}
\right], 
$
and let 
$\cA'=\langle Y'_1, \dots, Y'_m, Z'\rangle$. 
%
It is easy to see that the $Y'_i$'s and $Z'$ can be obtained from the 
$Y_i$'s and $Z$
via simultaneous row and column operations.
Note that simultaneous row and column operations do not change 
the rank or the discrepancy of a space, that is 
$\mincork(\cA')=\cork(\cA)$ and $\disc(\cA')=\disc(\cA)$. 
Observe that $\mincork(\cA)=\mincork(\cB)$, and by Lemma~\ref{lem:whatever}  we 
have
$\disc(\cA)=\disc(\cB)$. Therefore taking some $\cB$ not 
in the Edmonds-Rado class (for example $\sk_3$) it follows that $\cA$ and 
$\cA'$ are not 
in the Edmonds-Rado class. 
To finish the proof just note that 
if $A=I$, then $Y'_i$'s and $Z'$ are projections,
and if $A$ is a positive matrix with entries at least  the absolute 
values of the entries in the $B_i$'s, then $Y'_i$'s and $Z'$ are positive 
matrices.  
\end{proof}

\subsection{Compression spaces}
If $\maxrk(\cB)$ is of primary interest, in analogy with the Edmonds-Rado class 
we can define the following matrix class. Here we allow non-square matrices 
from $M(n\times n', \F)$. Recall that for $A\in M(n\times n', \F)$, its rank is 
$\dim(\im(A))$, its corank is $\dim(\ker(A))$, and $A$ is nonsingular if 
$\rk(A)=\min(n, n')$. 

Following the terminology used in~\cite{eh88, fr04}, we call
a matrix space $\cB\leq M(n\times n', \F)$ a {\em compression space}
if $\cB$ possesses $\mincork(\cB)$-singularity  witnesses.
In terms of discrepancy, $\cB$ is
a compression space
if $\mincork(\cB)=\disc(\cB)$. 
Thus, by the result of Lov\'asz discussed in Subsection~\ref{subsec:previous},
and by the result of Atkinson and Stephens used in
Subsection~\ref{subsec:pep},
rank-one spanned matrix spaces as well as two-dimensional
matrix spaces over sufficiently large base fields are compression spaces.
As to Wong sequences, from Lemma~\ref{lem:second_utility} we immediately have 
that if $\cB\leq M(n\times n', \F)$ is 
a compression space,
then for 
any $A\in\cB$, $A$ is of maximum rank if and only if the limit of the second 
Wong sequence of $(A, \cB)$ is contained in $\im(A)$. 

It is clear that when $n=n'$, if $\cB$ is
a compression space 
then it 
is in the Edmonds-Rado class. The converse is not true. 
\begin{proposition}
\label{prop:inER_notQER}
There exists a matrix space in the Edmonds-Rado class which is not
a compression space. 
\end{proposition}
The proof of Proposition~\ref{prop:inER_notQER} relies on the following lemma, which 
also explains why we do not expect to achieve rank maximization for upper 
triangular matrices in Theorem~\ref{thm:main2}. 
\begin{lemma}\label{prop:reduction_to_triangular}
Rank maximization of matrix spaces can be reduced to rank maximization of 
matrix spaces with a basis of pairwise commuting, and strictly upper triangular 
matrices.
\end{lemma}
\begin{proof}
For $\cB=\langle B_1, \dots, B_m\rangle\leq M(n\times n', \F)$ we first pad 
$0$'s to make it a matrix space of $M(\max(n, n'), \F)$. Then consider the 
matrix space in $M(2\cdot \max(n, n'), \F)$ generated by $C_1, \dots, C_m$ 
where $C_i=\left[
\begin{array}{cc}
0 & B_i \\
0 & 0
\end{array}
\right]
$. 
\end{proof}
\begin{proof}[Proof of Proposition~\ref{prop:inER_notQER}]
Consider the following matrix space: apply the construction in 
Lemma~\ref{prop:reduction_to_triangular} with $\sk_3$, and let the resulting 
matrix space be $\cB\leq M(6, \Q)$. $\cB$ is in the Edmonds-Rado class as it is 
spanned by upper-triangular matrices. 
On the other hand $\cB$ is not in the Edmonds-Rado class as $\cork(\cB)=4$ 
while $\disc(\cB)=3$. 
\end{proof}


\subsection{The black-box Edmonds-Rado class}

\begin{definition}
Let $\cB\leq M(n\times n', \F)$. 
$\cB$ is in the black-box Edmonds-Rado class if the following two conditions 
hold: 
(1) there exists a $\cork(\cB)$-singularity witness; 
(2) for any $A\in\cB$, either $A$ is of maximum rank, or 
$\cB(\ker(A))\not\subseteq \im(A)$. 
\end{definition}
By the first condition, the black-box Edmonds-Rado class is a subclass of 
the compression spaces.
Also note that $\cB(\ker(A))$ is just the first item 
in the 
second Wong sequence of $(A, \cB)$. 
The second condition says that
if $A$ is non-maximum rank then already the first 
item in the second Wong sequence excludes existence of $\cork(A)$-singularity
witnesses. In this case for any matrix $B$ from $\cB$ with
$B(\ker(A))\not\subseteq \im(A)$, we have $\rk( B)>\rk( A)$. Therefore
in matrix spaces in this class the following simple algorithm
finds an element of maximum rank over sufficiently large base fields.

\begin{proposition}
Let $\cB\leq \matspc$ be in the black-box Edmonds-Rado class, and assume 
$|\F|=\Omega(n)$. Then there exists a deterministic algorithm that solves the 
constructive SMR for $\cB$ using polynomial number of arithmetic operations. 
\end{proposition}
\begin{proof}
Given $A\in \cB$, we compute the rank of $A+\lambda B$ where $\lambda$ is from 
a subset
of $\F$ of size $\rk (A)+1$ and $B$ is from a basis of $\cB$. If none
of these matrices have rank larger than $A$, conclude that $A$ is
of maximum rank. Otherwise replace $A$ with an $A+\lambda B$ of
larger rank. Iterate the above procedure to obtain $A\in\cB$ of maximum rank. 
\end{proof}

As a justification for the name of the subclass,
observe that this algorithm does not make use of any properties
of matrices other that their rank. It even works
in the setting that instead of inputting the basis 
$B_1,\ldots,B_m$ explicitly, we only know $m$ and
 have an oracle which, on input 
$(\alpha_1,\ldots,\alpha_m)$ returns the rank 
of $\alpha_1B_1+\ldots+\alpha_mB_m$. 

\subsubsection{Some matrix spaces in the black-box Edmonds-Rado class.} While 
this class seems quite restrictive, 
it contains some interesting cases. 

A first example is when $\cB$ has a basis of 
positive semidefinite matrices. 
Let $\cB=\langle B_1, \dots, B_m\rangle\leq M(n, \R)$, where $B_i$'s are 
positive semidefinite. Then it is seen easily that $A$ is of maximum rank if 
and only if $\ker(A)=\cap_{i\in[m]}\ker(B_i)$. In particular if $A$ is not of 
maximum rank then there exists $v\in\ker(A)$ 
such that $B_j(v)\not\in\im(A)$, for some $j\in[m]$. 

Another more interesting scenario is from \cite{CIK97} (see also \cite{IKS}, 
Lemma~4.2). Let $G$ be a finite dimensional
         associative algebra over $\F$ and let $V, V'$ be semisimple 
	$G$-modules.
       Let  $\cB=\Hom_G(V,V').$
	Recall that a semisimple module is the direct sum of
	simple modules and that in a semisimple module every
	submodule has a direct complement. We know that
        $A \in \cB$ is of maximum rank if and only if for every isomorphism
	type $S$ of simple modules for $A$, the multiplicity of 
	$S$ in $\im(A)$ is the minimum of the multiplicities of $S$ in 
	$U$ and $V$.

        If $A$ is not of maximum rank, then for some simple module $S$ there 
	is an isomorphic copy $S_1$ of
        $S$ in $\ker(A)$ and there is a copy $S_2$ of $S$ in $V'$ 
	intersecting $\im(A)$ trivially. Also, there are nontrivial
	homomorphisms
        mapping the first copy of $S$ to the second one. 
	For instance, any isomorphism
	$S_1\to S_2$ can be extended to a homomorphism $V\to V'$ by
	the zero map on a direct complement of $S_1$.

        On the other hand, if $A$ is of maximum rank then for every
        simple submodule in $\ker(A)$, the copies in $V'$ isomorphic
        to it are in $\im(A)$, therefore no simple constituent
        can be moved out of $\im(A)$ via the second Wong sequence.

\section{Concluding remarks}
\label{sec:concl}

Our main results are deterministic polynomial time algorithms for
the {\em constructive} version of Edmond's problem (that is, finding 
nonsingular matrices) in certain
subclasses of the 
Edmonds-Rado class.
In the light of Gurvits' result
on the non-constructive version, probably the most interesting open
problem is the deterministic complexity of the constructive version
for the whole Edmonds-Rado class. Regarding the Boolean complexity 
of some of our algorithms, the bottleneck is our limited knowledge 
about the possible blowup of the sizes of bases for the Wong sequences.
We are not even aware of any good bound on the size of bases for 
singularity witnesses (except for the rank one generated case).
In particular, we do not know the Boolean complexity of finding singularity 
witnesses for singular triangularizable matrix spaces over the
rationals.

The deterministic or randomized complexity of finding rank one 
matrices spanning a rank-one generated space is another question
which is open to our knowledge. We believe that the problem
is hard. In contrast, triangularizing a triangularizable matrix
space may be easier. In the special case when the space
is triangularizable over the base field $\F$ and it contains
a nonsingular matrix (which can be efficiently found even
deterministically), Lemma~\ref{lem:triang-reg} gives a
reduction to finding composition series for matrix algebras
which is further reducible to factorization of polynomials over $\F$.
It would also be interesting finding maximum rank matrices
over very small fields in the rank-one spanned case. The algorithm
of \cite{IKS} does the job when rank-one generators are at hand.

\paragraph{Acknowledgements.}
We
would like to thank the anonymous reviewers for careful reading and pointing 
out some gaps in an earlier version of the paper.
Most of
this work was conducted when G.~I., Y.~Q. and M.~S. were
at the Centre
for Quantum Technologies (CQT) in Singapore, and partially funded by
the Singapore Ministry of Education and the National Research
Foundation, also through the Tier 3 Grant ``Random numbers from quantum 
processes'' (MOE2012-T3-1-009).
Research partially supported by the European Commission
IST STREP project Quantum Algorithms (QALGO) 600700,
by the French ANR Blanc program under contract ANR-12-BS02-005 (RDAM project), 
by the Hungarian Scientific Research Fund (OTKA), 
and by the Hausdorff grant EXC59-1/2.
\bibliography{ref}

\end{document}